\documentclass[conference]{IEEEtran}
\usepackage{epsfig,amsmath,amssymb,epsf,cite,algorithm,algorithmic}
\usepackage{graphicx}
\usepackage{color}
\usepackage{stmaryrd}
\usepackage{amssymb}
\usepackage{tipa}
\usepackage{setspace}

\graphicspath{{figures/}}

\newtheorem{proposition}{Proposition}

\def\b0{{\pmb{0}}}

\newcommand{\IorbI}{I}

\begin{document}
\title{The Public Safety Broadband Network: \\ A Novel Architecture with Mobile Base Stations}

\author{
  \authorblockN{Xu Chen, Dongning Guo}
\authorblockA{Electrical Engineering \& Computer Science Department \\
    Northwestern University, 
    Evanston, IL, USA
\vspace{-1ex}
} \and
\authorblockN{John Grosspietsch}
\authorblockA{
Enterprise Mobility Solutions\\
Motorola Solutions, Inc., 
Schaumburg, IL, USA
\vspace{-1ex}}
}


\maketitle


\begin{abstract}
A nationwide interoperable public safety broadband network is being planned by the United States government. The network will be based on long term evolution (LTE) standards and use recently designated spectrum in the 700 MHz band. The public safety network has different objectives and traffic patterns than commercial wireless networks. In particular, the public safety network puts more emphasis on coverage, reliability and latency in the worst case scenario. Moreover, the routine public safety traffic is relatively light, whereas when a major incident occurs, the traffic demand at the incident scene can be significantly heavier than that in a commercial network. Hence it is prohibitively costly to build the public safety network using conventional cellular network architecture consisting of an infrastructure of stationary base transceiver stations. A novel architecture is proposed in this paper for the public safety broadband network. The architecture deploys stationary base stations sparsely to serve light routine traffic and dispatches mobile base stations to incident scenes along with public safety personnel to support heavy traffic. The analysis shows that the proposed architecture can potentially offer more than 75\% reduction in terms of the total number of base stations needed.
\end{abstract}


\IEEEpeerreviewmaketitle

\section{Introduction}
\label{sec:intro}

The Federal Communications Commission (FCC) has recently designated
2$\times$10 MHz paired spectrum in the 700 MHz band for exclusive public safety uses nationwide. A unified public safety broadband network is being planned based on long term evolution (LTE) technologies.
The planned public safety network and existing commercial networks have different
characteristics.
Commercial networks have a higher user equipment (UE) density and target on a high
network revenue by providing diverse levels of services.
 The public safety network, however, aims at providing immediate access to the network (an excessive delay may cause loss of life and property), reliable communications with
guaranteed throughput and quality of service (QoS), as well as full nationwide coverage.

The traffic pattern of the public safety network and that of
commercial networks are also different. There are
two typical sources of traffic in the public safety network, namely, the light traffic due to routine activities such as patrols and
surveillance, and the heavy traffic due to large numbers of public
safety personnel at a major incident scene.
There have been some studies of public safety networks. A model for public safety traffic under normal
and emergency scenarios has been developed in \cite{CPCMI09}. Reference \cite{SCLST04} analyzed the public safety traffic on land
mobile radio systems and found close fitting distributions for call inter-arrival time and call
holding times. In \cite{CSVVT07}, it is shown that the average number of busy channels in most
cells is small compared to their capacities, whereas there are busy periods of high
utilization.
Averaged over time and across areas, the amount of routine traffic in a public safety network is much lower than that in a typical commercial network. When a major incident occurs,
however, the amount of traffic reaches its peak, where the aggregate
demand in a cell is at least comparable to that in a typical cell in a commercial network.

The FCC estimates the cell site density and network deployment cost of the public safety network to be comparable to a commercial network \cite{PJAP10}, while some have estimated it to be even a few times more than deploying a commercial network due to more stringent requirements \cite{RJ10}.
%
%

In this paper, we propose a novel architecture for
building an economic nationwide
public
safety broadband network. The wireless access points of the network consist of
sparsely deployed stationary base transceiver stations
(BTSs) for supporting light routine traffic and a distributed set of mobile BTSs ready to be deployed quickly to any incident
scene by vehicle or helicopter. A premise of this architecture is that
a mobile BTS can be dispatched to the incident scene as quickly as a
large number of personnel and can be set up quickly to provide needed wireless
services. This imposes a requirement on the density and placement of
mobile BTSs, as well as technologies that link the mobile BTSs to the
fixture infrastructure through, e.g., a wireless backhaul.

The proposed architecture is compared with the conventional architecture, where the cell sites are designed to satisfy the throughput requirement due to the sum of both light routine traffic and heavy incident scene traffic. We consider the worst case scenario where the incident scene is located at the cell edge. A stylized case study is carried out to obtain insights on how the throughput of the public safety network depends on the density of BTSs or the cell size under each architecture. In the conventional architecture, it is found that the cell size is mainly constrained by the heavy incident scene traffic. For the proposed architecture, depending on the traffic patterns, the backhaul and routine traffic demands are usually the critical constraints on the cell size. When the emergency traffic is mostly local, e.g., only 50\% of the incident scene traffic traverses the backhaul link to the core network, the routine traffic demand becomes the key constraint. It is found that the maximum cell size allowed using the proposed architecture to meet the traffic requirements is much larger than that using the conventional architecture. In all, the proposed architecture provides significant reduction in the total number of BTSs needed to provide the same throughput.

The remainder of the paper is organized as follows. Section
\ref{sec:architect} describes the conventional and the proposed
architectures. Section \ref{sec:analysis} analyzes the throughput for
downlink transmissions under those architectures. A case study 
through simulation is provided in Section \ref{sec:simresults}. Section \ref{sec:conclusion} concludes the paper.

\section{Network Architectures}
\label{sec:architect}


There are two types of traffic in the public safety network. One is
the light traffic induced by routine activities, such as patrols and
surveillance. The other is heavy traffic due to a large number of public safety personnel at a major incident scene. The network architecture must be designed to meet the peak capacity requirement due to the sum of both types of traffic.

\subsection{The Conventional Architecture}

The conventional cellular network architecture is based on an infrastructure of stationary BTSs connected to the wired core network. If such an architecture is adopted by the nationwide public safety network, the deployment of BTSs needs to be dense enough to meet the peak demand of public safety UEs, including both routine UEs and incident scene UEs. The cost is estimated to be a few times more than deploying a commercial network \cite{PJAP10}. The utilization of the network is low most of the time, because the routine traffic is much lower than the peak traffic.

\subsection{The Proposed Architecture}

The novel architecture proposed here adapts to the traffic demand by
using a collection of mobile BTSs in addition to 
a sparse network of
stationary BTSs deployed to support light routine traffic.
The mobile BTSs are sparsely distributed across the country, e.g.,
at selected fire stations. When a major incident occurs, a mobile BTS
is dispatched to the incident scene along with public safety
personnel.\footnote{It is conceivable to dispatch multiple mobile BTSs which coordinate to provide higher capacity. This is, however, out of the scope of the paper.} If the mobile BTS can start its service quickly (the initialization and configuration can begin en route), it suffices to have a mobile BTS at the scene just before the personnel at the scene becomes large.

Intuitively, the mobile BTS at the incident scene would be serving UEs
in a much smaller area than the coverage of a stationary BTS. Thus it
can support high data rate with enhanced QoS in the local area. One or
multiple wireless backhaul links are needed to connect the mobile BTS
to some stationary BTSs and thence to the core network. The
requirement on the capacity of the backhaul depends on how much of the
traffic at the incident scene is internal between local UEs, which
does not need to traverse the backhaul. Moreover, the backhaul link is
expected to be much stronger, because the mobile BTS remains mostly
stationary at the scene, has higher power and taller and more antennas
than the UEs, and that there may be a line of sight propagation.

\section{Throughput Analysis}
\label{sec:analysis}

We only consider downlink transmissions and we expect similar results to hold for uplink transmissions. It is assumed that every subcarrier of every link is subject to independent Rayleigh block fading with unit variance.

\subsection{A Generic Throughput Result}

We first give a generic result, which can then be applied to obtain
the throughputs of the conventional architecture and the proposed architecture in subsequent subsections.

We focus on a particular subcarrier and let $S_a$ denote the transmit
power spectral density of BTS $a$ on this subcarrier.
For this subcarrier of any link between two terminals $a$ and $b$, let $g(a,b)$ denote
the path loss and $h(a,b)$ denote the power gain of the Rayleigh
fading channel, which follows the exponential distribution with unit
mean, i.e., for every $x\ge0$,
\begin{align} \label{eq:distexp}
P\left\{ h(a,b)>x \right\} = e^{-x}.
\end{align}
The signal-to-interference-plus-noise ratio (SINR) of UE $u$ from
its associated base station $o$
is thus
\begin{align}
{\rm SINR}_u = \frac{S_o h(o,u) g(o,u) }{\eta + \sum\limits_{a \in O} S_{a}
g(a,u) h(a,u) },
\end{align}
where $O$ is the set of interfering BTSs
and $\eta$ is the noise power spectral density.
Let $\gamma$ be a prescribed threshold.
We say UE $u$'s spectral efficiency in the subcarrier (in bits/s/Hz) is
\begin{align}
\log(1+\gamma) \mathsf{P}\{ {\rm SINR}_u > \gamma \}.
\end{align}

\begin{proposition}\label{prop:throughput}
Let BTS $o$ serve all UEs in set $Z$ with bandwidth $W_o$ and transmit
power spectral density $S_o$. Each interfering BTS $a$ in set $O$ has
a transmit power spectral density $S_a$. The transmit power spectral
density is flat over its support. Suppose round-robin scheduling is
used, then the aggregate downlink throughput of all UEs in set $Z$ is
\begin{align}
\nonumber R &= W_o \log (1+\gamma) \frac{1}{|Z|} \sum\limits_{u \in Z} \Bigg[ \exp\left(-\frac{\eta \gamma}{g(u,o) S_o} \right) \\
& \hspace{2 cm} \prod\limits_{a \in O} \left(\frac{g(u,a) S_a \gamma}{g(u,o) S_o} +1\right)^{-1} \Bigg]
\end{align}
in bits per second (bps).
\end{proposition}
\begin{proof}
Suppose in a given time slot, a specific subcarrier is assigned to UE $u \in Z$. The total received interference power spectral density on UE $u $ is expressed as
\begin{align}
  \label{eq:Iu}
  \IorbI_u = \sum\limits_{a \in O} g(u,a) h(u,a) S_a.
\end{align}
Hence the received SINR of UE $u$ in the subcarrier is
\begin{align}
  \label{eq:1}
  {\rm SINR}_u = \frac{h(u,o) g(u,o) S_o }{  \eta  + \IorbI_u  }. 
\end{align}
Given $\IorbI_u$, which is independent of the channel gain $h(u,o)$,
the success probability is
\begin{align}
 \mathsf{P} \left\{ {\rm SINR}_u > \gamma \right\}
 &= \mathsf{E} \left\{    \mathsf{P}\left\{ {\rm SINR}_u > \gamma | I_u \right\} \right\} \\
  &= \mathsf{E} \left\{   \mathsf{P}\left\{ h(u,o) > \frac{\gamma (\eta + \IorbI_u) }{g(u,o) S_o}  \bigg| \IorbI_u \right\} \right\}\label{eq:pfstep3}\\
  &= \mathsf{E} \left\{  \exp\left(-\frac{\gamma  (\eta + \IorbI_u) }{g(u,o) S_o} \right)  \right\} \label{eq:expI},
\end{align}
where \eqref{eq:expI} is due to \eqref{eq:distexp} and the expectation
is taken over $I_u$.

Since $I_u$ is a sum of independent random variables, the success probability can be further written as
\begin{align}
\nonumber \mathsf{P} &\left\{ {\rm SINR}_u > \gamma \right\} \\
\nonumber  & = \exp\left(-\frac{\eta \gamma}{g(u,o) S_o} \right)  \\
\label{eq:pfstep4} & \hspace{1cm} \times \prod\limits_{a \in O} \mathsf{E} \left\{ \exp \left( - \frac{ h(u,a) g(u,a) S_a \gamma}{g(u,o) S_o}   \right) \right\} \\
\label{eq:pfstep5} & =  \exp\left(-\frac{\eta \gamma}{g(u,o) S_o} \right)
\prod\limits_{a \in O}
\left(\frac{g(u,a) S_a \gamma}{g(u,o) S_o} + 1\right)^{-1},
\end{align}
where~\eqref{eq:pfstep5} is by definition of the Laplace transform of
the exponential distribution with unit mean:
\begin{align}
  \mathsf{E} \{ e^{-s h(u,a)} \} = (s+1)^{-1}. 
\end{align}

Round-robin scheduling allocates the subcarrier to each UE in
$Z$ with a fraction of $1/|Z|$ of the time.  
By symmetry, the total throughput on all subcarriers is given by
\begin{align}
\label{eq:pfstep2} R =
W_o \frac{1}{|Z|} \sum\limits_{u \in Z} \log(1+\gamma)
\mathsf{P}\left\{ {\rm SINR}_u > \gamma \right\} .
\end{align}
Plugging in~\eqref{eq:pfstep5} proves the proposition.
\end{proof}

\subsection{Throughput of the Conventional Architecture}

In the conventional architecture, a stationary BTS serves both the
routine UEs and the incident scene UEs.
The total downlink bandwidth is $W$.  The total transmit power of every
stationary BTS
is $P$.
Denote the stationary BTS covering the incident scene by $o$.
Denote the set of routine UEs and incident scene UEs served by the
same BTS by $U$ and $C$, respectively.
Each 
UE experiences intercell interference from all BTSs from set $O$.

Applying Proposition \ref{prop:throughput} with
$|Z| = |U| +|C|$,
$W_o = W$,
and
$S_o = S_a = {P}/{W}$,
 the aggregate throughput of routine UEs is 
\begin{align}
\tilde{R}_u = \frac{W \log(1+\gamma) }{|U|+|C|} \sum\limits_{u \in U} \exp\left(-\frac{\eta W \gamma}{g(u,o) P} \right)
\prod\limits_{a \in O}
\frac1{\frac{g(u,a) \gamma}{g(u,o)} +1}
\end{align}
and the aggregate throughput of incident scene UEs is
\begin{align}
\tilde{R}_c = \frac{W \log(1+\gamma) }{|U|+|C|} \sum\limits_{c \in C} \exp\left(-\frac{\eta W \gamma}{g(c,o) P}\right)
\prod\limits_{a \in O }
\frac1{\frac{g(c,a) \gamma}{g(c,o)}+1}.
\end{align}

\subsection{Throughput of the Proposed Architecture}

\begin{figure}
  \centering
  \includegraphics[width=\columnwidth]{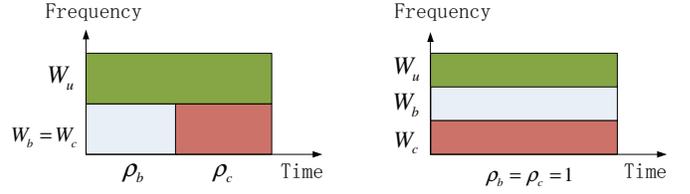}
  \caption{Left: TDRS; right: FDRS.}
  \label{fig:FT}
\end{figure}

Let the proposed architecture be adopted.
The total downlink bandwidth is $W$.
Let the bandwidth allocated to the link from the stationary BTS to
routine UEs be $W_u$.
The remaining bandwidth is used for transmissions by the stationary
BTSs to the mobile BTS (via the downlink wireless backhaul) and incident scene
UEs.
Let the mobile BTS use either frequency-domain resource sharing~(FDRS)
or
time-domain resource sharing~(TDRS)
to orthogonally divide the frequency and time resources for the two
types of links.
Let the bandwidth of the (downlink) backhaul be $W_b$ and
the bandwidth of the link from the mobile BTS to incident
scene UEs be $W_c$.
Further denote $\rho_b$ and $\rho_c$ as the proportion of time
resource allocated to the corresponding links, respectively.
See Fig.~\ref{fig:FT} for an illustration, where the parameters
satisfy the following relationships:
\begin{itemize}
\item Under TDRS, $W_b = W_c = W- W_u$ and $\rho_b + \rho_c = 1$.
\item Under FDRS, $W_b + W_c = W - W_u$ and $\rho_b = \rho_c = 1$.
\end{itemize}

Let the total transmit power of a stationary BTS and a mobile BTS be $P$
and $P^{\prime}$, respectively.
The transmit power from a stationary BTS to a mobile BTS associated
with it
is $P_b < P$.
The transmit power from a stationary BTS to its routine UEs is thus $P-P_b$.

To avoid interference to the mobile BTS located at the cell edge, we
disallow all stationary BTSs to transmit when the backhaul link is
active under TDRS or to transmit on the frequency resources allocated
to the backhaul under FDRS. The throughput to routine UEs in all cells
is degraded by no more than the degradation suffered by the routine
UEs in the cell in which the incident occurs. Hence, all cells satisfy
the routine traffic requirement if the cell covering the incident
scene does.

Let $o^{\prime}$
denote the mobile BTS at the center of the incident scene.
Applying $S_o = {P_b }/{ W_b}$, $W_o = W_b$, $S_a =0$ and $Z =
o^{\prime}$ to Proposition~\ref{prop:throughput} and taking into
account the fraction of time $\rho_b$ allocated to the backhaul, the
backhaul throughput is
\begin{align}
R_b = \rho_b W_b \log(1+\gamma)\exp\left(-\frac{\eta W_b \gamma}{g(o,o^{\prime}) P_b}\right)
\end{align}
where the parameters $\rho_b$ and $W_b$ depend on whether FDRS or TDRS
is used.

We next calculate the throughput to the UEs.
In the case of TDRS, applying $W_o = W_u$, $S_o = (P - P_b)/ W_u$, $S_a = {P}/{W}$ and $Z = U$ in Proposition \ref{prop:throughput}, the aggregate throughput from the stationary BTS to routine UEs (in bps) is
\begin{align}
\nonumber
R_u^{\rm TDRS} & = \frac{W_u}{|U|} \log(1+\gamma) \sum\limits_{u \in U} \exp\left(-\frac{\eta W_u \gamma}{g(u,o) (P-P_b)}\right) \\
\label{eq:RuPropArch}&
\qquad
\prod\limits_{a \in O }
\left(\frac{g(u,a) P W_u \gamma}{g(u,o) (P-P_b) W}  +1 \right)^{-1}.
\end{align}
Similarly, applying $W_o = W_c$, $ S_o ={ P^{\prime} }/{ W_c}$,
$S_a = {P}/{W}$ and $Z = C$ to Proposition~\ref{prop:throughput}
and taking into account the fraction of time $\rho_c$, the aggregate
throughput from the mobile BTS to its incident scene UEs (in bps) is
\begin{align}
\nonumber R_c^{\rm TDRS} &= \frac{\rho_c W_c}{|C|} \log(1+\gamma) \sum\limits_{c \in C} \exp\left(-\frac{\eta W_c \gamma}{g(c,o^{\prime}) P^{\prime} } \right) \\
 & \qquad \prod\limits_{a \in O }
\left( \frac{g(c,a) \gamma P W_c}{g(c,o^{\prime}) P^{\prime} W} +1 \right)^{-1}.
\end{align}

In the case of FDRS, applying $W_o = W_u$, $S_o = (P - P_b)/{ W_u}$, $S_a
= {P}/(W-W_b)$ and $Z = U$ to Proposition~\ref{prop:throughput}, the
aggregate throughput from the stationary BTS to routine UEs (in bps)
is
\begin{align}
\nonumber R_u^{\rm FDRS} &= \frac{W_u}{|U|} \log(1+\gamma) \sum\limits_{u \in U} \exp\left(-\frac{\eta W_u \gamma}{g(u,o) (P-P_b)}\right) \\
\label{eq:RuPropArch}& \,\,\,\quad \prod\limits_{a \in O }
\left(\frac{g(u,a) P W_u \gamma}{g(u,o) (P-P_b) (W-W_b) }   + 1\right)^{-1}.
\end{align}
Applying $W_o = W_c$, $ S_o = {P^{\prime} }/{ W_c}$, $S_a = {P}/(W -
W_b)$ and $Z = C$ to Proposition~\ref{prop:throughput}, the aggregate
throughput from the mobile BTS to its incident scene UEs (in bps) is
\begin{align}
\nonumber R_c^{\rm FDRS} &= \frac{W_c}{|C|} \log(1+\gamma) \sum\limits_{c \in C} \exp\left(-\frac{\eta W_c \gamma}{g(c,o^{\prime}) P^{\prime} } \right) \\
 & \qquad \prod\limits_{a \in O }
\left( \frac{g(c,a) \gamma P W_c}{g(c,o^{\prime}) P^{\prime} (W -
    W_b)}  + 1\right)^{-1}.
\end{align}

\section{A Case Study via Simulation}
\label{sec:simresults}

\begin{figure}
\centering
  \includegraphics[width=8cm]{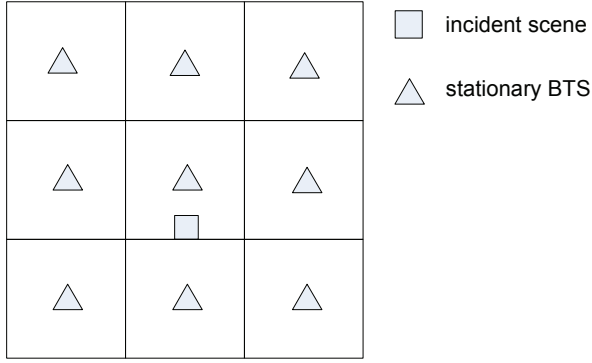}
  \caption{Network topology with square cell site and the incident scene being at the cell edge.}\label{fig:topology}
\end{figure}

\begin{table}
\small
\centering
\caption{System parameters.}
\begin{tabular}{|l|c|}
  \hline
  Incident scene size & $200 m \times 200 m$ \\ \hline
  Number of routine UEs in a cell & 80 \\ \hline
  Number of incident scene UEs & 50 \\ \hline
  Path loss from stationary BTS & $34.5 + 35 \log_{10}(d)$, \\
   to mobile BTS &  $d$ in meters \\ \hline
  Path loss from BTS  & $39.3 + 37.6 \log_{10} (d)$, \\
   to routine UEs/incident UEs & $d$ in meters \\ \hline
   Total number of resource blocks & 50 \\ \hline
   Number of subcarriers/resource block & 12 \\ \hline
   Subcarrier spacing & 15 kHz \\ \hline
   Transmit power of a stationary BTS $P$ & 46 dBm \\ \hline
   Transmit power for a backhaul $P_b$ & 45 dBm \\ \hline
   Transmit power of a mobile BTS $P^{\prime}$ & 43 dBm \\ \hline
   Noise power spectral density $\eta$ & -174 dBm/Hz \\ \hline
   SINR threshold $\gamma$ & 10 dB \\ \hline
\end{tabular}
\label{tab:simpara}
\end{table}
\normalsize

\begin{figure}[t]
  \centering
  \includegraphics[width=\columnwidth]{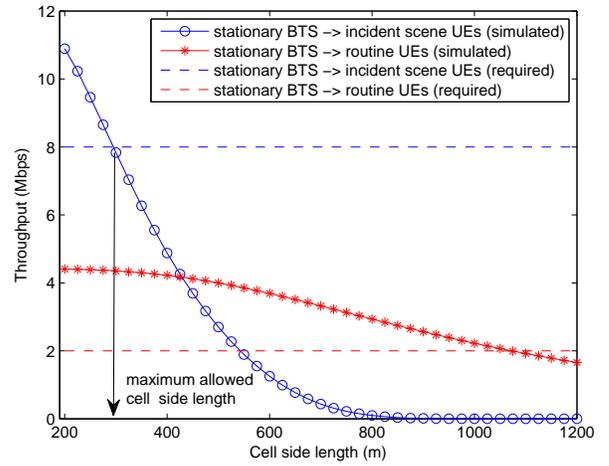}
  \caption{Tradeoff between routine UE throughput and incident scene UE throughput for the conventional architecture.}\label{fig:throughputconvarch}
\end{figure}

Consider a stylized network topology as illustrated in
Fig.~\ref{fig:topology}, where the public safety network is composed
of square grids and an incident occurs at the cell edge. The routine
UEs and incident scene UEs are evenly spaced according to lattices in
each cell and the incident scene, respectively. The system parameters
chosen according to LTE standards~\cite{LTEStandard} are listed in
Table~\ref{tab:simpara}.
Throughout the simulation, the total downlink bandwidth is $W = 50 \times 12
\times 15$ kHz $= 9$ MHz.


We compare the number of BTSs needed by the conventional architecture
and the proposed architecture, respectively,
so that the aggregate throughput required for the routine UEs, the
incident scene UEs and the backhaul are $T_u = 2$ Mbps, $T_c = 8 $
Mbps and $T_b = 4$ Mbps, respectively. In this case, $4/8 = 50\%$ of
the emergency traffic traverses the backhaul.
Fig.~\ref{fig:throughputconvarch} plots the throughput as a function of
the cell side length if the conventional architecture is adopted.
The side length must not exceed 300 $m$ in order to
meet the required throughput for the incident scene UEs ($T_c$).

Figs.~\ref{fig:throughputproparch} and~\ref{fig:throughputproparchFDRS} show the throughput tradeoff
 for the proposed architecture under TDRS and FDRS,
 respectively.
Under TDRS, the number of resource blocks
 assigned to a backhaul link is 25, i.e., $W_u = W_b = W_c = 25 \times
 12 \times 15 \text{ kHz} = 4.5 \text{ MHz}$ ($W=9$ MHz) and $\rho_b = 0.4$.
Under FDRS, the number of resource blocks assigned to the backhaul and
the link from the mobile BTS to incident scene UEs is 10 and 15,
 respectively.  Hence $W_u = 4.5 \text{ MHz}$, $W_b = 1.8
 \text{ MHz}$ and $W_c = 2.7 \text{ MHz}$.
It can be seen in both figures 
that the routine
UEs' rate ($T_u$) puts the critical constraint on the cell size. It
suffices to have 900~$m$ side length for each square
cell. Consequently, the number of stationary BTS required is reduced
by slightly more than $1- (300/900)^2 = 8/9$ using the proposed
architecture.

The remaining crucial question is how many mobile BTSs are needed in
the entire network. Here we give a crude estimate by assuming that a
mobile BTS can be dispatched at the same time as the first responders
and arrive at the incident scene within three times the amount of time
it takes the nearest fire engine to arrive.  The assumption is
reasonable because the stationary BTSs can support a small number of
personnel at the beginning of an incident, and a mobile BTS is needed
only when a much larger number of personnel arrive.  Under this assumption, the number of mobile BTSs required is about 1/9 of the total fire stations across the country (approximately 48,800 of them according to the National Fire Department census). The number of mobile BTSs needed is thus about $48,800/9 \approx5,422$.  The total number of stationary BTSs in the U.S. is about 44,000 in 2010 according to the FCC \cite{PJAP10}.  The total number of stationary and mobile base stations required by the proposed architecture is thus about $1/9 + 5422/44000 \approx 23\%$ as many as that required by the conventional architecture.

A mobile BTS is quite different from a stationary BTS.
In this simulation, the transmit power of a mobile BTS is only one-half of
that of a stationary BTS.
A mobile BTS may be portable. It does not require any building roof
top space, but needs to be transported to the incident scene
quickly. In remote areas where deployment of cellular infrastructure
is not cost effective, or in disasters where the cellular
infrastructure is destroyed, mobile BTSs can perhaps be helicoptered
in to provide crucial local access points, and backhaul links could be
established via satellites. Taking these into account, the proposed
architecture may bring an even larger reduction on the total cost of
the public safety broadband network.

\begin{figure}
  \centering
  \includegraphics[width=\columnwidth]{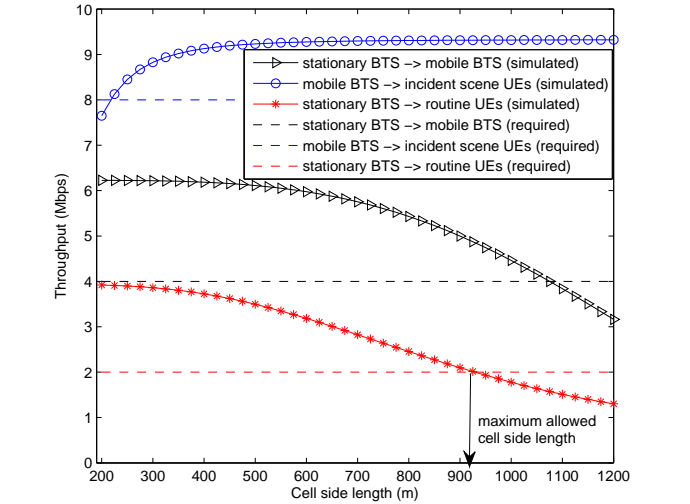}
  \caption{TDRS between backhaul link and the mobile BTS to incident scene UEs link. Tradeoff between (1) routine UE throughput, (2) incident scene UE throughput and (3) backhaul throughput for the proposed architecture. }\label{fig:throughputproparch}
\end{figure}

\begin{figure}
  \centering
  \includegraphics[width=\columnwidth]{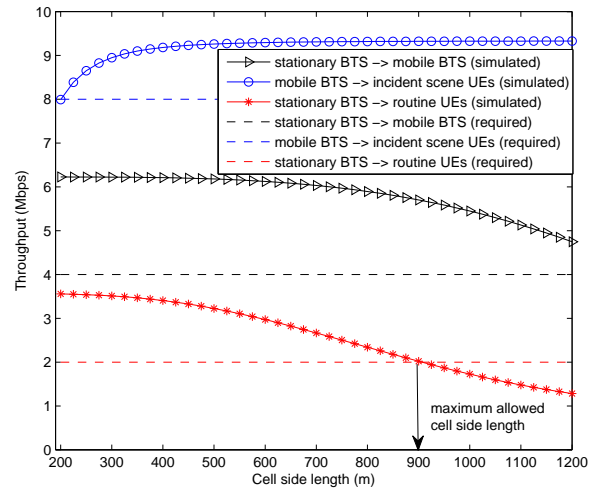}
  \caption{FDRS between backhaul link and the mobile BTS to incident scene UEs link. Tradeoff between (1) routine UE throughput, (2) incident scene UE throughput and (3) backhaul throughput for the proposed architecture. }\label{fig:throughputproparchFDRS}
\end{figure}

\section{Conclusion}
\label{sec:conclusion}

In this paper, we have proposed a novel network architecture for the public safety broadband network, where stationary BTSs are deployed sparsely and a mobile BTS is dispatched to the scene when a major incident occurs. A stylized case study shows that utilizing the mobile BTS is a promising solution, especially if the emergency traffic is mostly local to the incident scene. The proposed architecture can be a viable solution for the nationwide public safety broadband network.

\bibliographystyle{IEEEtran}
\bibliography{IEEEabrv,architectpsnbib}

\begin{thebibliography}{1}
\providecommand{\url}[1]{#1}
\csname url@samestyle\endcsname
\providecommand{\newblock}{\relax}
\providecommand{\bibinfo}[2]{#2}
\providecommand{\BIBentrySTDinterwordspacing}{\spaceskip=0pt\relax}
\providecommand{\BIBentryALTinterwordstretchfactor}{4}
\providecommand{\BIBentryALTinterwordspacing}{\spaceskip=\fontdimen2\font plus
\BIBentryALTinterwordstretchfactor\fontdimen3\font minus
  \fontdimen4\font\relax}
\providecommand{\BIBforeignlanguage}[2]{{%
\expandafter\ifx\csname l@#1\endcsname\relax
\typeout{** WARNING: IEEEtran.bst: No hyphenation pattern has been}%
\typeout{** loaded for the language `#1'. Using the pattern for}%
\typeout{** the default language instead.}%
\else
\language=\csname l@#1\endcsname
\fi
#2}}
\providecommand{\BIBdecl}{\relax}
\BIBdecl

\bibitem{CPCMI09}
C.~Chen, C.~Pomalaza-Raez, M.~Colone, R.~Martin, and J.~Isaacs, ``Modeling of a
  public safety communication system for emergency response,'' in \emph{Proc.
  IEEE International Conference on Communications}, Dresden, Germany, June
  2009, pp. 1 --5.

\bibitem{SCLST04}
D.~Sharp, N.~Cackov, N.~Laskovic, Q.~Shao, and L.~Trajkovic, ``Analysis of
  public safety traffic on trunked land mobile radio systems,'' \emph{IEEE
  Journal on Selected Areas in Communications}, vol.~22, no.~7, pp. 1197 --
  1205, Sept. 2004.

\bibitem{CSVVT07}
N.~Cackov, J.~Song, B.~Vujicic, S.~Vujicic, and L.~Trajkovic, ``Simulation and
  performance evaluation of a public safety wireless network: Case study,''
  \emph{Simulation}, vol.~81, no.~8, pp. 571--585, August 2005.

\bibitem{PJAP10}
J.~M. Peha, W.~Johnston, P.~Amodio, and T.~Peters, ``The public safety
  nationwide interoperable broadband network: A new model for capacity,
  performance and cost,'' \emph{FCC White Paper}, June 2010.

\bibitem{RJ10}
\BIBentryALTinterwordspacing
R.~Hallahan and J.~M. Peha, ``Quantifying the costs of a nationwide public
  safety wireless network,'' \emph{Telecommunications Policy}, vol.~34, no.~4,
  pp. 200 -- 220, 2010. [Online]. Available:
  \url{http://www.sciencedirect.com/science/article/pii/S0308596110000133}
\BIBentrySTDinterwordspacing

\bibitem{LTEStandard}
\emph{3GPP TR 36.931 version 9.0.0 Release 9}, May 2011.

\end{thebibliography}
\end{document}